\documentclass{amsart}
\usepackage{amssymb,bbm}
\date{September 5, 2020}

\def\cz{\mathbb{C}}
\def\rz{\mathbb{R}}
\def\nz{\mathbb{N}}

\def\bs{\mathbb{S}}

\def\dx{\mathrm{d}x}
\def\domega{d\omega}

\def\gB{\mathfrak{B}}

\def\rd{\mathrm{d}}

\newtheorem{corollary}{Corollary}
\newtheorem{lemma}{Lemma}
\newtheorem{proposition}{Proposition}
\newtheorem{theorem}{Theorem}

\newcommand{\1}{\mathbbm{1}}

\renewcommand{\epsilon}{\varepsilon}

\DeclareMathOperator{\tr}{Tr}

\begin{document}

\title[Relativistic Strong Scott Conjecture]{Relativistic Strong Scott
  Conjecture:\\ A Short Proof}
 
\author[R.~L. Frank]{Rupert L. Frank} \address{Mathematisches Institut,
  Ludwig-Maximilans Universit\"at M\"unchen, Theresienstr. 39, 80333
  M\"unchen, Germany, and Munich Center for Quantum Science and
  Technology (MCQST), Schellingstr. 4, 80799 M\"unchen, Germany, and
  Mathematics 253-37, Caltech, Pasadena, CA 91125, USA}
\email{rlfrank@caltech.edu}

  \author[K. Merz]{Konstantin Merz} \address{Institut f\"ur Analysis
    und Algebra, Carolo-Wilhelmina, Universit\"atsplatz 2, 38106
    Braunschweig} \email{k.merz@tu-bs.de}

  \author[H. Siedentop]{Heinz Siedentop} \address{Mathematisches
    Institut, Ludwig-Maximilans Universit\"at M\"unchen,
    Theresienstr. 39, 80333 M\"unchen, Germany, and Munich Center for
    Quantum Science and Technology (MCQST), Schellingstr. 4, 80799
    M\"unchen, Germany} \email{h.s@lmu.de}

\begin{abstract}
  We consider heavy neutral atoms of atomic number $Z$ modeled with
  kinetic energy $(c^2p^2+c^4)^{1/2}-c^2$ used already by
  Chandrasekhar.  We study the behavior of the one-particle ground
  state density on the length scale $Z^{-1}$ in the limit
  $Z,c\to\infty$ keeping $Z/c$ fixed.  We give a short proof of a
  recent result by the authors and Barry Simon showing the convergence
  of the density to the relativistic hydrogenic density on this scale.
\end{abstract}

\maketitle

\section{Introduction\label{s1}}

A simple description exhibiting some qualitative features of atoms of
large atomic number $Z$ with $N$ electrons and with $q$ spin states
each is offered by the Chandrasekhar operator
\begin{equation}
  \sum_{\nu=1}^N\left(\sqrt{-c^2\Delta_\nu+c^4}-c^2-\frac{Z}{|x_\nu|}\right)
  +\sum_{1\leq \nu<\mu\leq N}\frac{1}{|x_\nu-x_\mu|}
  \quad\text{in}\ \bigwedge_{\nu=1}^N L^2(\rz^3:\cz^q)
\end{equation}
where $c$ denotes the velocity of light.  It is defined as the
Friedrichs extension of the corresponding quadratic form with form
domain $\bigwedge_{\nu=1}^NC_0^\infty(\rz^3:\cz^q)$. By Kato's
inequality \cite[Chapter 5, Equation (5.33)]{Kato1966}, it follows
that the form is bounded from below if and only if $Z/c\leq2/\pi$ (see
also Herbst \cite[Theorem 2.5]{Herbst1977} and Weder
\cite{Weder1974}).  For $Z/c<2/\pi$ its form domain is
$H^{1/2}(\rz^{3N}:\cz^{q^N})\cap \bigwedge_{\nu=1}^N
L^2(\rz^3:\cz^q)$.  Since there is no interaction involving the
electron's spin present, we set $q=1$ for notational
simplicity. Moreover, we are restricting ourselves to neutral atoms
$N=Z$ and fix $\gamma=Z/c\in (0,2/\pi)$. We denote the resulting
Hamiltonian by $C_Z$.

In the following, we are interested in properties of ground states of
this system. Lewis et al \cite{Lewisetal1997} proved that the ground
state energy is an eigenvalue of $C_Z$ belonging to the discrete
spectrum of $C_Z$. Given any orthonormal base of the ground state space
$\psi_1,\ldots,\psi_M$ any ground state of $C_Z$ can be written as
$$
\sum_{\mu=1}^M w_\mu\lvert\psi_\mu\rangle\langle\psi_\mu\rvert
$$
where $w_\mu\geq 0$ are weights such that $\sum_{\mu=1}^M w_\mu=1$. In
the following, we would not even need a ground state. States which
approximate the ground state energy sufficiently well would be
enough. However, we refrain from such generalizations and simply pick
the state that occurs according to L\"uders \cite{Luders1951} when
measuring the ground state energy, namely the one with equal weights
$w_1=\cdots=w_M=M^{-1}$.  We write $d_Z$ for this state. Its
one-particle ground state density is
$$
\rho_Z(x) := N\sum_{\mu=1}^Mw_\mu\int_{\rz^{3(N-1)}} |\psi_\mu(x,x_2,\ldots,x_N)|^2\,\dx_2 \cdots \dx_N.
$$
For $\ell\in\nz_0$ we denote by $Y_{\ell,m}$, $m=-\ell,\ldots,\ell$, a
basis of spherical harmonics of degree $\ell$, normalized in
$L^2({\mathbb{S}}^2)$ \cite[Formula (B.93)]{Messiah1969} and by
\begin{equation}
  \label{pl}
  \Pi_{\ell} =\sum_{m=-\ell}^\ell|Y_{\ell,m}\rangle\langle Y_{\ell,m}|
\end{equation}
the projection onto the angular momentum channel $\ell$.  The electron
density $\rho_{\ell,Z}$ of the L\"uders state in the $\ell$-th angular
momentum channel is
\begin{multline}
  \rho_{\ell,Z}(x) \\
  := {N\over4\pi} \sum_{\mu=1}^M w_\mu \sum_{m=-\ell}^\ell \int_{\rz^{3(N-1)}} \!\!|\int_{\bs^2}\overline{Y_{\ell,m}(\omega)}\psi_\mu(|x|\omega,x_2,\ldots,x_N)\domega|^2 \dx_2\cdots\dx_N. 
\end{multline}
We note the relation
\begin{equation}
  \label{eq:totalaverage}
  \rho_Z
  = \sum_{\ell=0}^\infty \rho_{\ell,Z}.
\end{equation}

Our main result concerns these densities on distances of order
$Z^{-1}$ from the nucleus. We recall that electrons on these length
scales are responsible for the Scott correction in the asymptotic
expansion of the ground state energy
\cite{Solovejetal2008,Franketal2008} and are described by the
Chandrasekhar hydrogen Hamiltonian
$\sqrt{-\Delta+1}-1 - \gamma|x|^{-1}$ in $L^2(\rz^3)$. Spherical
symmetry leads to the radial operators
\begin{equation}
  \label{eq:defclh}
  C_{\ell,\gamma} := \sqrt{-\frac{\mathrm d^2}{\mathrm d r^2}+\frac{\ell(\ell+1)}{r^2}+1}-1-\frac{\gamma}{r}
\end{equation}
in $L^2(\rz_+):=L^2(\rz_+,\rd r)$.  We write $\psi_{n,\ell}^H$,
$n\in\nz_0$, for a set orthonormal eigenfunctions of $C_{\ell,\gamma}$
spanning its pure point spectral space. The hydrogenic density in
channel $\ell$ is
\begin{equation}
  \label{eq:2.12}
  \rho_{\ell}^H(x):= (2\ell+1)\sum_{n=0}^\infty |\psi_{n,\ell}^H(|x|)|^2/(4\pi|x|^2),
  \ x\in\rz^3,\end{equation}
and the total hydrogenic density is then given by
\begin{equation}
  \label{eq:totaldenshydro}
  \rho^H := \sum_{\ell=0}^\infty \rho_{\ell}^H.
\end{equation}
The generalization of Lieb's strong Scott conjecture \cite[Equation
(5.37)]{Lieb1981} to the present situation asserts the convergence of
the rescaled ground state densities $\rho_Z$ and $\rho_{\ell,Z}$ to
the corresponding relativistic hydrogenic densities $\rho^H$ and
$\rho_\ell^H$. Whereas the non-relativistic conjecture was proven by
Iantchenko et al \cite{Iantchenkoetal1996}, it was shown in the
present context in \cite{Franketal2019P}, including the convergence of
the sums defining the limiting objects $\rho^H_\ell$ and $\rho^H$.
Here, we will take the existence of the limiting densities for
granted. The purpose of this note is to offer a simpler proof of the
core of the above convergence of the quantum densities. It is a simple
virial type argument which allows us to obtain the central estimate on
the difference of perturbed and unperturbed one-particle Chandrasekhar
eigenvalues in an easy way. In addition we will refrain from using the
elaborate classes of test functions of \cite{Franketal2019P}.

\begin{theorem}[Convergence for fixed angular momentum]
  \label{convfixedl}
  Pick $\gamma\in(0,\frac2\pi)$, $\ell_0\in\nz_0$, assume
  $Z/c=\gamma$ fixed and $U:\rz_+\to\rz$ with
  $C_{\ell,-\gamma}^{-1/2}\circ U\circ C_{\ell,-\gamma}^{-1/2}\in \gB(L^2(\rz_+))$.
  Then
  \begin{equation}
    \label{konvergenz1}
    \lim_{Z\to\infty}\int_{\rz^3}c^{-3}\rho_{\ell_0,Z}(c^{-1}x)U(|x|)\,\dx
    = \int_{\rz^3} \rho_{\ell_0}^H(x)U(|x|)\,\dx.
  \end{equation}
\end{theorem}
We would like to add three remarks:

1. Our hypothesis allows for Coulomb tails of $U$ in contrast to
\cite[Theorem 1.1]{Franketal2019P} where the test functions were assumed to
decay like $O(r^{-1-\epsilon})$.

2. Since $C_{\ell,0}\geq C_{0,0}$ the Sobolev inequality shows that
$U$ such that $U\circ |\cdot|\in L^3(\rz^3)\cap L^{3/2}(\rz^3)$ is
allowed.

3. Picking a suitable class of test functions $U$, we may -- by an
application of Weierstra{\ss}' criterion -- sum \eqref{konvergenz1}
and interchange the limit $Z\to \infty$ with the sum over
$\ell_0$. This yields for $\gamma\in (0,2/\pi)$ and $Z/c=\gamma$ fixed
the convergence of the total density
  \begin{equation}
    \label{konvergenz2}
    \lim_{Z\to\infty}\int_{\rz^3} c^{-3}\rho_Z(c^{-1}x)U(|x|)\,\dx = \int_{\rz^3} \rho^H(x)U(|x|)\,\dx.
  \end{equation}
We refer to \cite{Franketal2019P} for details.

\section{Proof of the convergence}
\label{sec:strategy}

The general strategy of the proof of Theorem \ref{convfixedl} is a
linear response argument which was already used by Baumgartner
\cite{Baumgartner1976} and Lieb and Simon \cite{LiebSimon1977} for the
convergence of the density on the Thomas-Fermi scale and by Iantchenko
et al \cite{Iantchenkoetal1996} and \cite{Franketal2019P} in the
context of the strong Scott conjecture.  The rescaled density
integrated against a test function $U$ is written by the
Hellmann-Feynman theorem as a derivative of the energy with respect to
perturbing one-particle potential $-\lambda U$ and taking the limit
$Z\to \infty$ as the representation
\begin{equation}
  \label{grund}
\int_{\rz^3} \frac1{c^3}\rho_{\ell,Z}(x/c)U(|x|)\,\dx
= \frac{1}{\lambda c^2}\tr\{[C_Z-(C_Z-\lambda \sum_{\nu=1}^Z (U_c\otimes\Pi_{\ell})_\nu)]d_Z\}
\end{equation}
where $U_c(r):=c^2U(cr)$.
Of course, it is enough to prove \eqref{konvergenz2} for positive $U$, since
we can simply prove the result for the positive and negative part
separately and take the difference.  By standard estimates following
\cite{Franketal2019P} (which in turn are patterned by
Iantchenko et al \cite{Iantchenkoetal1996}) one obtains
\begin{proposition}
  \label{redux1}
  Fix $\gamma:=Z/c\in(0,2/\pi)$, $\ell\in\nz_0$, and assume that
  $U\geq 0$ is a measurable function on $(0,\infty)$ that is form
  bounded with respect to $C_{\ell,\gamma}$, and assume that $|\lambda|$ is
  sufficiently small. Then
  \begin{equation}
    \label{folgerung1}
    (2\ell+1)\sum_n{e_{n,\ell}(0)-e_{n,\ell}(\lambda)\over \lambda}
    \begin{cases}
       \geq  \limsup\limits_{Z\to\infty} \int_{\rz^3} \rho_{\ell,Z}(c^{-1}x) U(|x|)\,\dx&  \text{if}\ \lambda>0\\
      \leq  \liminf\limits_{Z\to\infty} \int_{\rz^3} \rho_{\ell,Z}(c^{-1}x)
      U(|x|)\,\dx&  \text{if}\ \lambda<0
    \end{cases},
  \end{equation}
  where $e_{n,\ell}(\lambda)$ is the $n$-th eigenvalue of $C_{\ell,\gamma}-\lambda U$.
\end{proposition}

Thus the limit \eqref{konvergenz1} exists, if the derivative of the
sum of the eigenvalues with respect to $\lambda$ exists at $0$, i.e.,
\begin{equation}
  \label{ableitung}
  \left.{\rd \over \rd \lambda}\sum_n e_{n,\ell}(\lambda)\right|_{\lambda=0}
\end{equation}
exists and -- when multiplied by $-(2\ell+1)$ -- is equal to the right
of \eqref{konvergenz1}. Put differently: the wanted limit exists, if
the Hellmann-Feynman theorem does not only hold for a single
eigenvalue $e_{n,\ell}(\lambda)$ of $C_{\ell,\gamma}-\lambda U$ but
for the sum of all eigenvalues. However, the differentiability would
follow immediately, if we were allowed to interchange the
differentiation and the sum in \eqref{ableitung}, since the
Hellmann-Feynman theorem is valid for each individual nondegenerate
eigenvalue and yields the wanted contribution to the derivative.  In
turn, the validity of the interchange of these two limiting processes
would follow by the Weierstra{\ss} criterion for absolute and uniform
convergence, if we had in a neighborhood of zero a
$\lambda$-independent summable majorant of the moduli of the summands
on the left side of \eqref{folgerung1}. This, in turn is exactly the
content of Lemma \ref{changeinev} enabling to interchange the limit
$\lambda\to0$ with the sum over $n$ and concluding the proof of
Theorem \ref{convfixedl}.

Before ending the section we comment on the difference to previous
work: Already Iantchenko et al \cite[Lemma 2]{Iantchenkoetal1996}
proved a bound similar to \eqref{eq:weierstrass} in the
non-relativistic setting where -- in contrast to the Chandrasekhar
case -- the hydrogenic eigenvalues are explicitly known. However, this
was initially not accessible in the present context. Instead, the
differentiability of the sum was shown in \cite[Theorems 3.1 and
3.2]{Franketal2019P} by an abstract argument for certain self-adjoint
operators whose negative part is trace class.  To prove the analogue
majorant for the Chandrasekhar hydrogen operator is the new
contribution of the present work yielding a substantial
simplification.

\section{The majorant}

Before giving the missing majorant, we introduce some useful
notations.  Set
\begin{equation}
  \label{A}
  A:=2+{2^{3/2}\over\pi(\sqrt2-1)}.
\end{equation}

For $\gamma\in(0,2/\pi)$ and $t\in[0,1)$ we set
\begin{equation}
  \label{F}
  F_\gamma(t):=(1-t) \left({1+t\over 1-t} \right)^{1+A}
  \left({\frac2\pi-\gamma\over \frac2\pi-{1+t\over 1-t}\gamma}\right)^{A}={(1+t)^{1+A}\over \left(1-{\frac2\pi+\gamma\over\frac2\pi-\gamma}t\right)^A}.
\end{equation}
Obviously $F_\gamma\in C^1([-t_0,t_0])$ with
$t_0:=(\frac1\pi-\frac\gamma2)/(\gamma+\frac2\pi)$. We set
\begin{equation}
  \label{M} 
  \tilde M_\gamma:=\max F_\gamma'([-t_0,t_0]).
\end{equation}
Furthermore, we write $C_\gamma$ for the optimal constant in the
following inequality \cite[Theorem 2.2]{Franketal2009} bounding all
hydrogenic Chandrasekhar eigenvalues from below by the corresponding
hydrogenic Schr\"odinger eigenvalues, i.e.,
\begin{equation}
  \label{eq:2.2}
  C_\gamma e_n( p^2/2-\gamma/|x|) \leq e_n(\sqrt{p^2+1}-1-\gamma/|x|).
\end{equation}
(Here, and sometimes also later, it is convenient to use a slightly
more general notation for eigenvalues of self-adjoint operators $A$
which are bounded from below: we write $e_0(A)\leq e_1(A)\leq \cdots$
for its eigenvalues below the essential spectrum counting their
multiplicity.)

Finally, we set
\begin{equation}
  \label{tM}
   M_\gamma:=C_\gamma \tilde  M_\gamma
 \end{equation}
 
 This allows us to formulate our central Lemma which allows to
 interchange the derivative in \eqref{ableitung} and proves Theorem
 \ref{convfixedl}. We recall that we write $e_{n,\ell}(\lambda)$ for
 the eigenvalues of $C_{\ell,\gamma}-\lambda U$ (see Proposition
 \ref{redux1}).
\begin{lemma}
  \label{changeinev}
  Assume $\gamma\in(0,2/\pi)$ and $U: \rz_+\to\rz_+$ such that the
  operator norm
  $b:=\|C_{\ell,-\gamma}^{-\frac12}UC_{\ell,-\gamma}^{-\frac12}\|$ is
  finite. Then for all $\lambda \in[-t_0/b,t_0/b]$ and all
  $\ell,n\in\nz_0$
   \begin{equation}
    \label{eq:weierstrass}
      \left| e_{n,\ell}(\lambda) - e_{n,\ell}(0) \right|
    \leq M_\gamma b |\lambda| \frac{\gamma^2}{(n+\ell+1)^2}.
  \end{equation}
\end{lemma}

For the proof, we need some preparatory results. We begin with a bound
on the change of Coulomb eigenvalues with the coupling constant which
is the core of our argument.
\begin{proposition}
  \label{main}
  For all $\gamma,\gamma'\in(0,2/\pi)$ with $\gamma\leq\gamma'$ and all
  $n\in\nz_0$
  \begin{equation}
    \label{eq:main1}
    e_n(\sqrt{p^2+1}-1-\gamma' |x|^{-1})
    \geq e_n(\sqrt{p^2+1}-1-\gamma|x|^{-1})\ \left({\gamma'\over\gamma}\right)^{1 +A}\left({\frac2\pi-\gamma\over\frac2\pi-\gamma'}\right)^A.
  \end{equation}
\end{proposition}
For the proof we will quantify the fact that eigenfunctions live essentially
in a bounded region of momentum space. 
\begin{lemma}
  \label{boundedmomentum}
  For all $\gamma\in(0,2/\pi)$ and all eigenfunctions $\psi$ of
  $\sqrt{p^2+1}-1-\gamma/|x|$
  \begin{equation}
    \label{uvschranke}
    \left\langle\psi,\left(\sqrt{p^2+1}-1\right)\psi\right\rangle
    \leq {\frac2\pi A\over \frac2\pi-\gamma}\left\langle\!\psi,\left(\!1-(p^2+1)^{-1/2}\!\right)\psi\!\right\rangle.
  \end{equation}
\end{lemma}

\begin{proof}
  Let $\psi_>=\1_{\{|p|>1\}}\psi$ and $\psi_<=\1_{\{|p|\leq 1\}}\psi$. Then
  $$
  \left\langle\psi,\left(\sqrt{p^2+1}-1\right)\psi\right\rangle
  = \left\langle\psi_<,\left(\sqrt{p^2+1}-1\right)\psi_<\right\rangle
  + \left\langle\psi_>,\left(\sqrt{p^2+1}-1\right)\psi_>\right\rangle
  $$
  and
  $$
  \left\langle\!\psi,\left(\!1-(p^2+1)^{-1/2}\!\right)\psi\!\right\rangle
  = \left\langle\!\psi_<,\left(\!1-(p^2+1)^{-1/2}\!\right)\psi_<\!\right\rangle
  + \left\langle\!\psi_>,\left(\!1-(p^2+1)^{-1/2}\!\right)\psi_>\!\right\rangle.
  $$
  We have
  $$
  \left\langle\psi_<,\left(\sqrt{p^2+1}-1\right)\psi_<\right\rangle \leq \sqrt 2 \left\langle\psi_<,\left(1-(p^2+1)^{-1/2}\right)\psi_<\right\rangle,
  $$
  since
  $ \sup_{0\leq e\leq 1} (\sqrt{e+1}-1)(1-1/\sqrt{e+1})^{-1} =
  \sup_{0\leq e\leq 1} \sqrt{e+1} = \sqrt 2$.

  Moreover, by Kato's inequality,
  \begin{equation*}
    \left\langle\psi_>,\left(\sqrt{p^2+1}-1\right)\psi_>\right\rangle
    \leq \left\langle\psi_>,|p|\psi_>\right\rangle \leq \frac{2/\pi}{2/\pi-\gamma} \left\langle\psi_>,\left(|p|-\gamma |x|^{-1} \right)\psi_>\right\rangle.
  \end{equation*}
  Now using the eigenvalue equation for $\psi$ we obtain
  \begin{align*}
    \left\langle\psi_>,\left(|p|-\gamma |x|^{-1} \right)\psi_>\right\rangle
    & = \left\langle\psi_>,\left(|p|-\gamma |x|^{-1} \right)\psi\right\rangle +\gamma \left\langle\psi_>,|x|^{-1}\psi_<\right\rangle \\
    & = \left\langle\psi_>,\left(\! E+|p|-\sqrt{p^2+1}+1 \!\right)\psi\right\rangle +\gamma \left\langle\psi_>,|x|^{-1}\psi_<\right\rangle \\
    & = \left\langle\psi_>,\left(\! E+|p|-\sqrt{p^2+1}+1 \!\right)\psi_>\right\rangle +\gamma \left\langle\psi_>,|x|^{-1}\psi_<\right\rangle \\
    & \leq 2 \left\langle\psi_>,\left(1-(p^2+1)^{-1/2}\right)\psi_>\right\rangle +\gamma \left\langle\psi_>,|x|^{-1}\psi_<\right\rangle.
  \end{align*}
  In the last inequality we used $E\leq 0$ (which was shown by Herbst
  \cite[Theorem 2.2]{Herbst1977}) and
  $ \sup_{e\geq 1} (\sqrt e- \sqrt{e+1}+1)(1-1/\sqrt{e+1})^{-1} = 2.$
  
  Moreover, by Hardy's inequality
  $$
  \langle\psi_>,|x|^{-1}\psi_<\rangle
  \leq \|\psi_>\| \||x|^{-1}\psi_<\| \leq 2 \|\psi_>\| \||p|\psi_<\| \leq \|\psi_>\|^2 + \||p|\psi_<\|^2.
  $$
  Since
  $$
  \|\psi_>\|^2 \leq \frac{1}{1-1/\sqrt 2}
  \left\langle\psi_>,\left(1-(p^2+1)^{-1/2}\right)\psi_> \right\rangle
  $$
  and
  $$
  \||p|\psi_<\|^2 \leq \frac{1}{1-1/\sqrt 2}
  \left\langle\psi_<,\left(1-(p^2+1)^{-1/2}\right)\psi_< \right\rangle,
  $$
  we can now collect terms and get
  \begin{align}
    \label{sammlung}
    \left\langle\psi,\left(\sqrt{p^2+1}-1\right)\psi\right\rangle 
    \leq  &\left(\sqrt 2 +{\sqrt2\frac2\pi\gamma\over(\sqrt2-1)(\frac2\pi-\gamma)}\right)\left\langle\psi_<,\left(1-(p^2+1)^{-1/2}\right)\psi_<\right\rangle\nonumber\\
    &+{\frac2\pi\over\frac2\pi-\gamma}\left(2 +{\sqrt2\gamma\over\sqrt2-1} \right)\left\langle\psi_>,\left(1-(p^2+1)^{-1/2}\right)\psi_>\right\rangle\\
    \leq &{\frac2\pi\over\frac2\pi-\gamma}\left(2 +{\sqrt2\gamma\over\sqrt2-1} \right)\left\langle\psi,\left(1-(p^2+1)^{-1/2}\right)\psi\right\rangle\nonumber
      \end{align}
  which gives the desired bound, since $\gamma<2/\pi$.
\end{proof}

\begin{corollary}
  \label{boundedmomentumcor}
  Let $\gamma\in(0,2/\pi)$ and $A$ be the constant of the
  previous lemma. Then, for any normalized
  eigenfunction $\psi$ of $\sqrt{p^2+1}-1-\gamma/|x|$ with
  eigenvalue $E$
  \begin{equation}
    \label{l1}
  \langle \psi, {\gamma\over|x|}\psi\rangle \leq
  \left({\frac2\pi A\over\frac2\pi-\gamma}+1\right)|E|.
  \end{equation}
\end{corollary}

\begin{proof}
  With the abbreviation $D:=\frac2\pi A(\frac2\pi-\gamma)^{-1}$ we write the inequality in the previous lemma in the form
  \begin{equation}
    \label{l2}
    \left\langle\psi,\frac{p^2}{\sqrt{p^2+1}}\psi\right\rangle
    \geq (1+D^{-1})
    \left\langle\psi,\left(\sqrt{p^2+1}-1\right)\psi\right\rangle.
  \end{equation}
  By the virial theorem (Herbst \cite[Theorem 2.4]{Herbst1977}) the
  left sides of \eqref{l1} and \eqref{l2} are equal. Thus,
  \begin{align*}
    &\langle \psi, \gamma|x|^{-1}\psi\rangle
     = (D +1) \langle \psi, \gamma|x|^{-1}\psi\rangle - D \left\langle\psi,\frac{p^2}{\sqrt{p^2+1}}\psi\right\rangle \\
    \leq &(D +1) \langle \psi, \gamma|x|^{-1}\psi\rangle - D (1+D^{-1}) \left\langle\psi,\left(\sqrt{p^2+1}-1\right)\psi\right\rangle \\
     = &- (D+1)\left\langle \psi,\left(\sqrt{p^2+1}-1-\gamma |x|^{-1}\right)\psi\right\rangle  = -(D+1) E
  \end{align*}
  as claimed.
\end{proof}

We are now in position to give the
\begin{proof}[Proof of Proposition \ref{main}]
  By the variational principle, for any $n\in\nz_0$ the function
  $\kappa\mapsto e_n(\sqrt{p^2+1}-1-\kappa|x|^{-1})$ is Lipschitz and
  therefore differentiable almost everywhere. By perturbation theory,
  at every point where its derivative exists, it is given by
  $$
  \frac{\rd}{\rd\kappa} e_n(\sqrt{p^2+1}-1-\kappa|x|^{-1}) =
  - \langle\psi_\kappa,|x|^{-1}\psi_\kappa\rangle \,,
  $$
  where $\psi_\kappa$ is a normalized eigenfunction of
  $\sqrt{p^2+1}-1-\kappa|x|^{-1}$ corresponding to the eigenvalue
  $e_n(\sqrt{p^2+1}-1-\kappa|x|^{-1})$.  Thus, by Corollary
  \ref{boundedmomentumcor}, we have for all $\kappa\in(0,\gamma']$
  $$
  \frac{\rd}{\rd\kappa} e_n(\sqrt{p^2+1}-1-\kappa|x|^{-1}) \geq
  \left({\frac2\pi A\over\frac2\pi-\kappa}+1\right) \kappa^{-1}
  e_n(\sqrt{p^2+1}-1-\kappa|x|^{-1}).
  $$
  Thus,
  \begin{align}
    \label{eq:derivativeev}
    \frac{\rd}{\rd\kappa}\log |e_n(\sqrt{p^2+1}-1-\kappa|x|^{-1})|
    \leq {A+1\over\kappa}+{A\over\frac2\pi-\kappa}.
  \end{align}
  Integrating this bound we find for $\gamma\leq\gamma'$ that
  $$
  \log \frac{|e_n(\sqrt{p^2+1}-1-\gamma'|x|^{-1})|}{|e_n(\sqrt{p^2+1}-1-\gamma|x|^{-1})|}
  \leq (A+1)\log\frac{\gamma'}{\gamma} -A\log{\frac2\pi-\gamma'\over\frac2\pi-\gamma},
  $$
  i.e.,
  $$
  e_n(\sqrt{p^2+1}-1-\gamma'|x|^{-1}) \geq
  e_n(\sqrt{p^2+1}-1-\gamma|x|^{-1})\ \left( {\gamma'\over\gamma}
  \right)^{A+1}\left({\frac2\pi-\gamma\over\frac2\pi-\gamma'}\right)^{A}$$
  quod erat demonstrandum.
\end{proof}

Eventually we can address the
\begin{proof}[Proof of Lemma \ref{changeinev}]
  First let $0<\lambda\leq t_0/b$.  Then
  $$
  \sqrt{p_\ell^2+1}-1-\frac\gamma r-\lambda U \geq (1-b\lambda)
  \left( \sqrt{p_\ell^2+1}-1-{1+b\lambda\over
      1-b\lambda}\cdot{\gamma\over r} \right)
  $$
  and therefore for all $n\in\nz_0$,
  $$
  e_n\left(\sqrt{p_\ell^2+1}-1-\frac\gamma r-\lambda U\right) \geq
  (1-b\lambda) \ e_n\left(
    \sqrt{p_\ell^2+1}-1-{1+b\lambda\over1-b\lambda}\cdot\frac\gamma r
  \right).
  $$
  By Proposition \ref{main} with
  $\gamma'=\gamma (1+b\lambda)/(1-b\lambda)$ (which fulfills
  $\gamma<\gamma'\leq \frac1\pi+\frac\gamma2<\frac2\pi$ under our
  assumptions)
  \begin{align*}
 & e_n\left( \sqrt{p_\ell^2+1}-1-{1+b\lambda \over 1-b\lambda}\cdot {\gamma\over r}\right)\\
    \geq &e_n\left( \sqrt{p_\ell^2+1}-1-\frac\gamma r\right) \cdot \left({1+b\lambda\over 1-b\lambda} \right)^{1+A}\left({\frac2\pi-\gamma\over \frac2\pi-{1+b \lambda\over 1-b\lambda}\gamma}\right)^{A}.
  \end{align*}
  Combining the previous two inequalities shows that
  \begin{equation*}
    e_{n,\ell}(\lambda)    \geq  F_{\gamma}(\lambda b) e_{n,\ell}(0)
  \end{equation*}
  and therefore, if $\lambda\leq t_0/b$,
  \begin{align*}
    & e_{n,\ell}(\lambda) - e_{n,\ell}(0)
    \geq (F_{\gamma}(\lambda b) - F_\gamma(0)) \ e_{n,\ell}(0)
    = \int_0^{\lambda b} F'_{\gamma}(t)\rd t \ e_{n,\ell}(0)\\
    \geq & \tilde M_{\gamma}\cdot b \cdot \lambda \cdot e_{n,\ell}(0)\geq - M_{\gamma}\cdot b \cdot \lambda  \cdot \frac{\gamma^2}{(n+\ell+1)^2}.
  \end{align*}
  In  the last inequality we used the lower bound \eqref{M}.

  The case of negative $\lambda$ is similar.
\end{proof}

\section*{Acknowledgments}

The authors warmly thank Barry Simon for his initial contributions and
continuing support and interest in the relativistic strong Scott
conjecture. They also acknowledge partial support by the U.S. National
Science Foundation through grants DMS-1363432 and DMS-1954995
(R.L.F.), by the Deutsche Forschungsgemeinschaft (DFG, German Research
Foundation) through grant SI 348/15-1 (H.S.) and through Germany's
Excellence Strategy – EXC-2111 – 390814868 (R.L.F., H.S.).
One of us (K.M.) would like to thank the organizers of the program
\textit{Density Functionals for Many-Particle Systems: Mathematical Theory and
  Physical Applications of Effective Equations}, which took place at the
Institute for Mathematical Sciences (IMS) at the
National University of Singapore (NUS), for their invitation to speak, their
kind hospitality, as well as for generous financial support by the Julian
Schwinger foundation that made his stay possible.


\def\cprime{$'$}

\end{document}